\documentclass[a4paper,UKenglish,cleveref,autoref,thm-restate]{lipics-v2021}

\pdfoutput=1 
\hideLIPIcs  

\title{CV-Rules: Serializability Verification of Concurrency Control Protocols via Explicit Transaction Ordering}

\titlerunning{CV-Rules: Serializability Verification via Explicit Transaction Ordering}

\author{Takashi Hoshino}{Cybozu Labs, Inc., Japan}{hoshino@labs.cybozu.co.jp}{https://orcid.org/0000-0003-4922-7479}{}
\author{Shigeo Mitsunari}{Cybozu Labs, Inc., Japan}{mitsunari@labs.cybozu.co.jp}{}{}
\author{Takashi Kambayashi}{NAUTILUS Technologies, Inc., Japan}{kambayashi@nautilus-technologies.com}{}{}
\author{Ryoji Kurosawa}{NAUTILUS Technologies, Inc., Japan}{kurosawa@nautilus-technologies.com}{}{}
\author{Sho Nakazono}{LY Corporation, Japan}{shnakazo@lycorp.co.jp}{}{}

\authorrunning{Hoshino et al.}
\Copyright{Takashi Hoshino, Shigeo Mitsunari, Takashi Kambayashi, Ryoji Kurosawa, Sho Nakazono}

\ccsdesc[500]{Information systems~Database transaction processing}
\ccsdesc[300]{Theory of computation~Logic and verification}
\ccsdesc[300]{Software and its engineering~Formal software verification}

\keywords{Serializability, Concurrency Control, Transaction Processing, Formal Verification, Theorem Proving}

\category{} 

\relatedversion{} 

\supplementdetails[cite={Hosh2026-CVRules-Artifact}]{Software}{https://doi.org/10.5281/zenodo.20755695}

\acknowledgements{} 

\nolinenumbers 

\EventEditors{}
\EventNoEds{2}
\EventLongTitle{}
\EventShortTitle{}
\EventAcronym{}
\EventYear{2026}
\EventDate{}
\EventLocation{}
\EventLogo{}
\SeriesVolume{}
\ArticleNo{}

\usepackage{cite} 
\usepackage{tabulary}
\usepackage{tikz}

\newcommand{\vf}{\mathit{vf}}
\newcommand{\Crule}{\textsf{C-rule}}
\newcommand{\Vrule}{\textsf{V-rule}}
\newcommand{\CVrules}{\textsf{CV-rules}}

\begin{document}

\maketitle

\begin{abstract}
We present CV-rules, an alternative characterization of serializability
in which a transaction order constructed by a protocol satisfies two
per-read conditions, C-rule (Causality) and V-rule (View
Consistency), that constrain the reads-from relation and
competing writers. While classical Multi-Version Serialization
Graph (MVSG) reasoning characterizes serializability via its acyclicity,
our approach requires explicit order construction, enabling
direct proofs that build on the protocol's own mechanisms.
We prove that CV-rules, serializability, and MVSG acyclicity are all equivalent.
Moreover, the C/V separation reveals that serializability is polynomial-time decidable for any fixed bound on the width of the order forced by C-rule.
We verify five protocols: Two-Phase Locking, Multi-Version Timestamp
Ordering, Serial Safety Net (SSN), Aria, and SnapChain. For SSN and Aria, whose original papers defined
only certification conditions, we identify explicit transaction orders
arising from their mechanisms; we also prove that Aria's unique-write constraint is
unnecessary for serializability. SnapChain, in contrast, is designed directly from CV-rules,
enforcing V-rule by construction.
All results except the complexity bounds are mechanized in Lean with no additional axioms and no admitted goals.
\end{abstract}

\section{Introduction}
\label{sec:introduction}

Serializability, the requirement that a concurrent execution be equivalent to some serial one, is the canonical correctness condition for transactional concurrency. Originating in database theory~\cite{Papa1979,Bern1983}, it underlies correctness reasoning across transaction processing, replicated data stores, and transactional memory. Verifying that concurrency control protocols guarantee serializability is a fundamental and recurring problem.

The classical approach to serializability verification uses the Multi-Version Serialization Graph (MVSG)~\cite{Bern1983}. A history is serializable if and only if some MVSG for it is acyclic. While theoretically elegant and well suited to protocols that directly construct dependency graphs, MVSG-based reasoning presents challenges for protocols that establish serializability through other mechanisms. The verification target is a global property (acyclicity of the entire dependency graph), and the serial order is a byproduct derived via topological sort, not an input specified by the protocol. Since order identification is not required, correctness proofs can establish acyclicity by contradiction without ever specifying the serial order. For protocols based on locks, timestamps, or certification conditions, these features of MVSG-based reasoning can obscure how protocol mechanisms actually establish the required ordering. This is particularly evident in recent protocols such as SSN~\cite{Wang2017-SSN} and Aria~\cite{Lu2020b-Aria}, whose original correctness proofs establish serializability via graph acyclicity without identifying which transaction order the protocol enforces.

We propose \CVrules{}, a characterization of serializability that operates on a transaction order \emph{explicitly constructed by the protocol}. Whereas MVSG identifies serializability with a global property (graph acyclicity), \CVrules{} identify it with two local per-read conditions: \Crule{} (Causality) constrains the reads-from relation, and \Vrule{} (View Consistency) constrains competing writers of the same data item.

\CVrules{} make protocol mechanisms visible. For SSN, we find that the protocol constructs a logical timestamp derived from anti-dependency analysis, a construction not explicit in the original work. For Aria, we show that the order combines a batch-local identifier with dependency-based reordering, and that Aria's unique-write constraint is unnecessary for serializability. The C/V separation further reveals the complexity structure of serializability: \Vrule{}'s disjunctive constraints drive its NP-hardness, and for any fixed bound on the width of the order forced by \Crule{}, the problem becomes polynomial-time decidable.

More broadly, \CVrules{} continue the declarative line of work on transactional consistency~\cite{Cero2015,Croo2017-Seeing,Xion2020}, which characterizes serializability and weaker consistency levels via axioms relating visibility, arbitration, or execution-test structures to an ordered history. Our framework specializes to serializability while allowing partial transaction orders, enabling direct application to concrete protocol mechanisms.

This paper makes the following contributions:
\begin{itemize}
\item \textbf{Theory.} We characterize serializability via two local per-read conditions (\Crule{} and \Vrule{}) and prove their equivalence with MVSG acyclicity. We also show that serializability is polynomial-time decidable for any fixed bound on the width of the order forced by \Crule{}.

\item \textbf{Protocol verification.} We verify five protocols (2PL, MVTO, SSN, Aria, and SnapChain) by explicitly constructing their transaction orders and proving \CVrules{}. This includes identifying implicit orders in existing protocols and, with SnapChain, using \CVrules{} as a design principle.

\item \textbf{Mechanization.} All theorems are mechanically verified in Lean (${\sim}11.2\text{K}$ lines across 105 modules) with no additional axioms and no admitted goals; only the complexity bounds are argued informally.
\end{itemize}

The remainder of the paper is organized as follows.
\cref{sec:background} presents background on multi-version histories and MVSG.
\cref{sec:cv-rules} defines \CVrules{}, proves their equivalence to serializability, and establishes bounded-width tractability.
\cref{sec:mvsg-equivalence} establishes the equivalence between \CVrules{} and MVSG and contrasts the two approaches.
\cref{sec:case-studies} applies \CVrules{} to verify protocols.
\cref{sec:mechanization} discusses the Lean mechanization.
\cref{sec:related-work} surveys prior work, and \cref{sec:conclusion} concludes.

\section{Background}
\label{sec:background}

This section presents the multi-version transaction model, defines serializability via serial execution, and recalls the classical MVSG characterization.

\subsection{Multi-Version Histories}
\label{sec:multi-version-histories}

A \emph{transaction} is a unit of work that reads and writes data items. We denote a read of data item $x$ by transaction $t_i$ as $r_i(x)$, and a write as $w_i(x)$. Following the \emph{two-step model}~\cite{Papa1984}, also adopted by Bernstein et al.~\cite{Bern1983}, each transaction reads and writes each data item at most once, and we assume transactions do not read from themselves (no self-reads). These restrictions are standard normal forms in multi-version serializability theory; histories with repeated accesses or self-reads can be reduced to this form without affecting serializability.

Formally, a history is defined as follows:

\begin{definition}[History]
A \emph{history} $h$ consists of:
\begin{itemize}
\item A finite set $D$ of data items.
\item A finite set $T$ of transactions.
\item For each transaction $t \in T$, a set of operations $\mathit{ops}(t) \subseteq \{r(x), w(x) \mid x \in D\}$ specifying which data items $t$ reads and writes.
\item A distinguished \emph{initialization transaction} $t_0 \notin T$ with $\mathit{ops}(t_0) = \{w(x) \mid x \in D\}$, i.e., $t_0$ writes all data items.
\item A \emph{version function} $\vf$ that, for each $t_i \in T$ and each $x \in D$ with $r(x) \in \mathit{ops}(t_i)$, specifies a unique transaction $\vf(x, t_i) \in T \cup \{t_0\}$ from which $t_i$ reads $x$.
\end{itemize}
\end{definition}

The version function captures the reads-from relation: $\vf(x, t_i) = t_j$ means transaction $t_i$ reads the version of $x$ written by $t_j$. Note that $\vf(x, t_i)$ is only defined when $r(x) \in \mathit{ops}(t_i)$, i.e., when $t_i$ actually reads $x$. We require that the source $t_j = \vf(x, t_i)$ actually writes $x$, i.e., $w(x) \in \mathit{ops}(t_j)$, and that $t_j \ne t_i$.

We call this a \emph{multi-version history} following the terminology of Bernstein et al.~\cite{Bern1983} (see also~\cite{Papa1984}), though our definition differs from the classical one. The key observation, also made by Bernstein et al.\ in their multi-version serializability theory, is that operation ordering within or across transactions does not affect serializability; only the version function matters. Accordingly, our definition omits operation orderings entirely, and the version function abstracts over both single-version and multi-version storage. This model deliberately abstracts away abort handling, operation-level scheduling, and real-time ordering: serializability here is the multi-version notion of Bernstein et al., determined solely by the version function.

\begin{example}
\label{ex:simple-history}
Consider three transactions and two data items:
\begin{align*}
t_1 &: w_1(x),\, w_1(y) \\
t_2 &: r_2(x),\, w_2(y) \\
t_3 &: r_3(x),\, r_3(y)
\end{align*}
A possible version function: $\vf(x, t_2) = t_1$, $\vf(x, t_3) = t_1$, $\vf(y, t_3) = t_2$. This indicates that both $t_2$ and $t_3$ read $x$ from $t_1$, while $t_3$ reads $y$ from $t_2$.
\end{example}

\subsection{Serializability via Serial Execution}
\label{sec:serializability}

Serializability requires that a concurrent execution be equivalent to some serial execution. We define this in terms of the version function.

Informally, a \emph{serial execution} runs transactions one at a time in some total order, where each transaction reads the most recent version of each data item, that is, the version written by the latest preceding writer. Given a total order $\prec$ on transactions, this naturally determines a version function: $\vf(x, t_i)$ is the maximum $t_j \prec t_i$ such that $t_j$ writes $x$. The following definition formalizes this idea.

\begin{definition}[Serializability]
\label{def:serializable}
A history $h$ is \emph{serializable} if there exists a total order $\prec$ on $T \cup \{t_0\}$ (with $t_0$ as the minimum) such that for every transaction $t_i \in T$ and every data item $x$ that $t_i$ reads:
\[
\vf(x, t_i) = \max_{\prec}\{t_j \mid t_j \prec t_i \text{ and } w(x) \in \mathit{ops}(t_j)\}
\]
(The right-hand side is well-defined: $t_0$ writes all data items and precedes every $t_i$ under~$\prec$, so $t_0$ is always in the candidate set.)
In other words, the version function matches what would result from executing transactions serially in that order.
\end{definition}

\begin{example}
Continuing \cref{ex:simple-history}, the serial order $t_0 \prec t_1 \prec t_2 \prec t_3$ produces exactly the version function given: $t_2$ and $t_3$ both read $x$ from the most recent writer $t_1$, and $t_3$ reads $y$ from $t_2$ (the most recent writer of $y$ before $t_3$). Thus this history is serializable.
\end{example}

\subsection{MVSG Characterization}
\label{sec:mvsg}

The Multi-Version Serialization Graph (MVSG)~\cite{Bern1983} provides a classical characterization of serializability. The construction requires a \emph{version order} $\ll$, which is a family $(\ll_x)_{x \in D}$ where each $\ll_x$ is a total order on the transactions that write $x$ (including $t_0$).

\begin{definition}[MVSG]
Given a history $h$ and a version order $\ll$, the \emph{Multi-Version Serialization Graph} $\mathit{MVSG}(h, \ll)$ is a directed graph where:
\begin{itemize}
\item Nodes are transactions in $T \cup \{t_0\}$.
\item Edges capture dependencies:
\begin{itemize}
\item \textbf{Initial}: $t_0 \to t_i$ for all $t_i \in T$.
\item \textbf{Reads-from}: $t_j \to t_i$ if $\vf(x, t_i) = t_j$ for some $x \in D$.
\item \textbf{Version-order}: $t_k \to t_j$ if $\vf(x, t_i) = t_j$, $w(x) \in \mathit{ops}(t_k)$, $t_k \notin \{t_i, t_j\}$, and $t_k \ll_x t_j$.
\item \textbf{Anti-dependency}: $t_i \to t_k$ if $\vf(x, t_i) = t_j$, $w(x) \in \mathit{ops}(t_k)$, $t_k \notin \{t_i, t_j\}$, and $t_j \ll_x t_k$.
\end{itemize}
\end{itemize}
\end{definition}

Intuitively, version-order edges enforce that older writers precede newer ones, while anti-dependency edges ensure that readers precede later writers.

The fundamental theorem of MVSG states:

\begin{theorem}[MVSG Theorem~\cite{Bern1983}]
\label{thm:mvsg}
A history is serializable if and only if there exists a version order $\ll$ such that $\mathit{MVSG}(h, \ll)$ is acyclic.
\end{theorem}

When the MVSG is acyclic, a topological sort yields a serial order. The MVSG characterization is protocol-agnostic: it applies regardless of how the protocol operates. However, the verification target is a global property (acyclicity of the entire graph), and the serial order is derived from this graph rather than given as input. For protocols that construct a transaction order through their own mechanisms, MVSG-based reasoning leads to two complications:

\noindent\textbf{Indirect ordering.} The serial order is derived from the graph structure via topological sort; it does not appear as an explicit input. When a protocol constructs an order directly, the correspondence between that order and the topological sort of the MVSG must be established separately.

\noindent\textbf{Indirect proofs.} Since the MVSG framework does not require order identification, correctness proofs can proceed by contradiction (assuming a cycle exists) without ever specifying the serial order. When the order a protocol enforces is non-obvious, proofs naturally follow this indirect strategy, establishing correctness without revealing which order the protocol constructs.

\section{\CVrules{}}
\label{sec:cv-rules}

This section introduces \CVrules{}, our characterization of serializability based on explicit transaction ordering.

\subsection{Definition}
\label{sec:cv-rules-definition}

\subsubsection{Transaction Order}

We extend the notion of history with an explicit transaction order:

\begin{definition}[Ordered History]
\label{def:ordered-history}
An \emph{ordered history} $(h, \prec)$ consists of a history $h$ and a strict partial order $\prec$ (irreflexive and transitive) on the transactions $T \cup \{t_0\}$, where $t_0$ is the minimum element: $t_0 \prec t$ for all $t \in T$.
\end{definition}

The order $\prec$ represents the intended serialization order. A key feature is that $\prec$ need only be a \emph{partial order}: we do not require all transactions to be comparable.

\subsubsection{\Crule{} (Causality)}

\Crule{} captures the basic causality constraint: if a transaction reads data from another transaction, the writer must precede the reader in the serialization order.

\begin{definition}[\Crule{}]
\label{def:c-rule}
An ordered history $(h, \prec)$ satisfies \Crule{} if for every transaction $t_i \in T$ and every data item $x$ that $t_i$ reads from $t_j$ (i.e., $\vf(x, t_i) = t_j$), we have $\boldsymbol{t_j \prec t_i}$.
\end{definition}

Intuitively, in any serial execution, a transaction can only read data that was written before it executed. \Crule{} directly encodes this causal dependency.

\subsubsection{\Vrule{} (View Consistency)}

\Crule{} requires the source to precede the reader, but not to be the \emph{most recent} writer. \Vrule{} fills this gap by constraining how other writers of the same item are positioned.

\begin{definition}[\Vrule{}]
\label{def:v-rule}
An ordered history $(h, \prec)$ satisfies \Vrule{} if for every transaction $t_i \in T$, every data item $x$ that $t_i$ reads from $t_j$ (i.e., $\vf(x, t_i) = t_j$), and every other transaction $t_k \notin \{t_i, t_j\}$ that writes $x$, we have $\boldsymbol{t_k \prec t_j}$ or $\boldsymbol{t_i \prec t_k}$.
\end{definition}

\begin{figure}[t]
\centering
\begin{tikzpicture}[
  dot/.style={circle, fill, inner sep=1.5pt},
  odot/.style={circle, draw, inner sep=1.5pt},
  lbl/.style={font=\footnotesize},
  arr/.style={->, semithick},
]

\draw[arr] (-0.3, 0) -- (9.5, 0) node[right, lbl] {$\prec$};

\fill[black!15] (3, -0.25) rectangle (6.5, 0.25);

\node[dot] (tj) at (3, 0) {};
\node[lbl, below=3pt] at (tj) {$t_j$};
\node[dot] (ti) at (6.5, 0) {};
\node[lbl, below=3pt] at (ti) {$t_i$};

\node[odot] (tk1) at (1.2, 0) {};
\node[lbl, below=3pt] at (tk1) {$t_k$};

\node[odot] (tk2) at (8.3, 0) {};
\node[lbl, below=3pt] at (tk2) {$t_k$};

\end{tikzpicture}
\caption{\Vrule{} illustrated.
  For a read $\vf(x, t_i) = t_j$,
  every competing writer $t_k$ of $x$
  must precede the source~$t_j$ or follow the reader~$t_i$.
  The shaded zone, including incomparable positions, is forbidden.}
\label{fig:v-rule}
\end{figure}
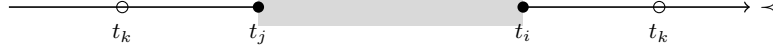

The intuition is that if $t_i$ reads $x$ from $t_j$, no other write to $x$ can interpose between $t_j$ and $t_i$ in the serialization order (\cref{fig:v-rule}). If $t_k$ wrote $x$ and appeared between $t_j$ and $t_i$, $t_i$'s view would break: serial execution would yield $t_k$'s version, not $t_j$'s. Note that since $\prec$ is a partial order, $t_k$ being concurrent with $t_j$ or $t_i$ (i.e., incomparable) is also prohibited; \Vrule{} requires an explicit ordering.

\begin{example}
Consider three transactions:
\begin{align*}
t_1 &: w_1(x) \\
t_2 &: w_2(x) \\
t_3 &: r_3(x) \text{ from } t_1
\end{align*}
\Vrule{} for the read by $t_3$ from $t_1$, considering the other writer $t_2$, requires:
$t_2 \prec t_1 \quad \text{or} \quad t_3 \prec t_2$.
This ensures $t_2$'s write does not appear between $t_1$ and $t_3$ in the serialization order. If $t_2 \prec t_1$, then $t_3$ correctly reads from $t_1$ (the most recent). If $t_3 \prec t_2$, then $t_2$'s write occurs after $t_3$ reads, so it does not break $t_3$'s view.
\end{example}

\subsubsection{Connection to MVTO}

\CVrules{} generalize the correctness conditions of the
Multi-Version Timestamp Ordering (MVTO) protocol~\cite{Reed1978-MVTO,Bern1983},
in which each transaction receives a timestamp and read/write
operations must respect the timestamp order.
MVTO's conditions are \Crule{}, \Vrule{}, and totality of
the timestamp order.
While the conditions themselves reference only local comparisons,
dropping totality makes \emph{sufficiency} non-trivial: one must
extend the partial order to a linear extension and use \Vrule{}'s
disjunctive structure to rule out intervening writers.
Necessity, in contrast, follows directly from any serial order
(\cref{thm:cv-rules}).
\CVrules{} thus characterize serializability and apply to
protocols whose mechanisms are substantially different from
timestamp ordering, as we demonstrate in \cref{sec:case-studies}.

\subsection{Main Theorem: Characterizing Serializability}
\label{sec:cv-rules-theorem}

We now state and prove the main theorem: \CVrules{} characterize serializability.

\begin{theorem}[\CVrules{} Equivalence]
\label{thm:cv-rules}
A history $h$ is serializable if and only if there exists a transaction order $\prec$ such that $(h, \prec)$ satisfies both \Crule{} and \Vrule{}.
\end{theorem}

\subsubsection{Proof Sketch}

\textbf{Sufficiency} (CV $\Rightarrow$ Serializable).
Given a partial order satisfying \Crule{} and \Vrule{}:
\begin{enumerate}
\item Apply Szpilrajn's extension theorem~(\cite{Szpi1930}; see also \cite[Theorem~10.2]{Schr2016-book}), explicitly via topological sorting, to extend the partial order to a total order. Since the extension preserves existing order relations, both rules carry over.
\item \Crule{} ensures $t_j \prec t_i$, and \Vrule{} then ensures that no other writer $t_k$ appears between them, so $t_j$ is the most recent writer before $t_i$. The version function thus matches serial execution.
\end{enumerate}

\noindent\textbf{Necessity} (Serializable $\Rightarrow$ CV).
Given a serializable history with serial order $\prec_s$:
\begin{enumerate}
\item Use $\prec_s$ as the serialization order.
\item \textbf{\Crule{}}: In serial execution, each read obtains the version from the most recent prior writer. Thus the source precedes the reader: \Crule{} holds.
\item \textbf{\Vrule{}}: For any read from $t_j$ by $t_i$ and any other writer $t_k$: if $t_j \prec t_k \prec t_i$, then $t_i$ would read $x$ from $t_k$, a contradiction. Hence $t_k \prec t_j$ or $t_i \prec t_k$: \Vrule{} holds.
\end{enumerate}

\begin{remark}[Complexity]
\label{rem:cv-rules-complexity}
Given an order~$\prec$, checking \CVrules{} runs in $O(|T|^2 \cdot |D|)$ time.
Determining whether \emph{any} satisfying order exists is NP-complete~\cite{Papa1979,Bern1983}, since by \cref{thm:cv-rules} this is equivalent to deciding serializability.
\end{remark}

\subsection{\Vrule{} Decomposition}
\label{sec:v-rule-decomposition}

\Vrule{} can be decomposed into two independent conditions.

\begin{theorem}[\Vrule{} Decomposition]
\label{thm:v-rule-decomposition}
Under \Crule{}, \Vrule{} is equivalent to the conjunction of:
\begin{enumerate}
\item \textbf{Endpoint comparability.}
For every read $\vf(x, t_i) = t_j$ and every other writer $t_k$
of~$x$, $t_k$ is comparable with both $t_i$ and $t_j$ under~$\prec$.
\item \textbf{No interposition.}
There is no such $t_k$ with $t_j \prec t_k \prec t_i$.
\end{enumerate}
\end{theorem}

\begin{proof}[Proof sketch]
\Vrule{}$\,\Rightarrow\,$(1)$\,\land\,$(2):
\Crule{} gives $t_j \prec t_i$; combining with \Vrule{}'s
disjunction and transitivity yields comparability.
(1)$\,\land\,$(2)$\,\Rightarrow\,$\Vrule{}:
if $t_k$ is comparable with both endpoints and does not interpose,
case analysis yields $t_k \prec t_j$ or $t_i \prec t_k$.
\end{proof}

This decomposition clarifies why partial orders suffice
for \CVrules{}.
Transactions on disjoint data items need not
be ordered.
The decomposition shows that even among writers of the
\emph{same} data item, not all pairs need be mutually comparable: endpoint
comparability only demands that each competing writer be comparable
with the two endpoints of the reads-from relationship it competes with.
In particular, \emph{non-visible writes}, writes whose versions
are never read~\cite{Thom1979,Naka2020-NWR}, need not be
ordered among themselves.

The decomposition also reflects how protocols ensure \Vrule{} in
practice.
Endpoint comparability often follows directly from the protocol's
ordering mechanism, while no-interposition depends on each protocol's
specific design choices: locks, version selection, or
certification (\cref{sec:case-studies}).

\subsection{Bounded-Width Tractability}
\label{sec:bounded-width}

As noted in \cref{rem:cv-rules-complexity}, determining serializability
is NP-complete in general.
The C/V separation pinpoints the source of this hardness:
the orderings forced by \Crule{} are uniquely determined
by the reads-from relation,
while \Vrule{} introduces \emph{disjunctive} constraints: for each
competing writer~$t_k$, either $t_k \prec t_j$ or $t_i \prec t_k$
must hold.
This disjunctive structure drives the NP-hardness.

When transaction concurrency is bounded by a fixed constant, however, the problem becomes
tractable.

\begin{definition}[Width]
\label{def:width}
The \emph{width} of a history is the maximum size of an antichain
(a set of pairwise incomparable transactions) under the partial order
forced by \Crule{}, that is, the transitive closure of the reads-from
relation $\{(t_j, t_i) \mid \exists x.\, \vf(x, t_i) = t_j\}$.
\end{definition}

\begin{theorem}[Bounded-width decidability]
\label{thm:bounded-width}
For any fixed~$k$, serializability of histories with width at
most~$k$ is decidable by a state-space search.
\end{theorem}

\begin{proof}[Proof sketch]
By Dilworth's theorem~(\cite{Dilw1950};
see also \cite[Theorem~2.26]{Schr2016-book}),
a partial order of width~$k$ can be partitioned into at most
$k$~chains.
A \emph{chain-progress vector} records, for each chain,
how far the serialization has progressed along that chain.
At each state, the only candidates for placement are the
at most~$k$ next unplaced transactions, one per chain,
and a candidate is eligible only when all its predecessors
in the forced ordering (including those in other chains) have
been placed.
Placing an eligible transaction advances the corresponding chain
and is accepted only if no \Vrule{} violation is introduced.
An exhaustive search from the initial vector (all chains at
position zero) explores all reachable states;
the history is serializable if and only if the search reaches
the final vector (all transactions placed).
Soundness follows because every accepted transition preserves
\Vrule{}, and completeness because the search explores
all valid placements.
\end{proof}

The state space contains $O(n^{k})$ chain-progress vectors
for~$n$ transactions, so for fixed~$k$ the procedure runs in
polynomial time.
Each transition evaluates at most~$k$ candidates and checks
\Vrule{} in $O(n \cdot |D|)$ time per candidate, yielding
$O(n^{k+1} \cdot |D|)$ total time for fixed~$k$.
Additional order constraints such as real-time or session order
can be added to the forced relation before measuring width;
for instance, session order bounds the width by the number of sessions.
Biswas and Enea~\cite{Bisw2019} showed that serializability is
decidable in polynomial time when the number of sessions is bounded,
also with $O(n^{k})$ state space.
Their result is recovered as the special case where session order is
included in the forced relation. Using the reads-from structure alone
yields a complementary parameterization, bounded even when session
count grows, as long as reads-from chains span sessions.
The C/V separation provides a structural explanation of
why bounded concurrency, under any parameterization, limits the
disjunctive constraints.

\subsection{Completeness}
\label{sec:completeness}

The necessity direction of \cref{thm:cv-rules} (Serializable $\Rightarrow$ CV) is straightforward: the serial order serves as the transaction order, and \Crule{} and \Vrule{} hold because the version function matches serial execution.
This has an immediate consequence for protocol verification.

\begin{proposition}[Serial executability implies completeness]
\label{prop:serial-executability}
In our version-function-only model (\cref{sec:multi-version-histories}), if a protocol can execute any sequence of transactions in a given serial order, then it is complete: every serializable history can be realized by the protocol.
\end{proposition}

Because our model captures only version functions, not operation ordering, completeness reduces to serial executability, a condition that most protocols satisfy trivially. This does not contradict classical results showing, e.g., that 2PL cannot produce all conflict-serializable schedules~\cite{Weik2002-book}: those limitations concern operation ordering. Indeed, all five protocols we verify (\cref{sec:case-studies}) satisfy serial executability and are thus complete and equivalent at the version-function level: each can realize every serializable history.

This equivalence does not imply interchangeability in practice, because our model is intentionally primitive: it characterizes serializability but does not distinguish protocols by performance. In classical theory, a protocol's \emph{scheduling space} (or \emph{scheduling power}~\cite[\S4.1]{Weik2002-book}) is the class of serializable histories it can realize at the operation level; our model abstracts this axis away. However, performance depends more directly on \emph{commit rate}, the fraction of transactions a protocol commits (rather than aborting and retrying) under a given workload, as demonstrated by restart-cost analysis~\cite{Kung1981-OCC}, contention modeling~\cite{Agra1987}, and workload-dependent evaluations~\cite{Yuxi2014-Abyss,Tana2020-CCBench}. Predicting commit rate requires modeling event arrival and abort strategies rather than operation ordering, an extension we leave to future work.

\section{Equivalence with MVSG}
\label{sec:mvsg-equivalence}

Having established that \CVrules{} characterize serializability, we now prove that \CVrules{} are equivalent to the classical MVSG characterization. This equivalence serves two purposes. First, it confirms that \CVrules{} capture the same notion of serializability as the established theory~\cite{Bern1983}: \CVrules{} do not define a new notion of correctness, but give an equivalent characterization via a different proof methodology. Second, it allows comparing the two approaches directly, which we do in \cref{sec:why-cv-rules}.

\subsection{Equivalence Theorem}
\label{sec:equivalence-theorem}

\begin{theorem}[MVSG Equivalence]
\label{thm:cv-mvsg-equivalence}
A history $h$ satisfies \CVrules{} (i.e., there exists a transaction order $\prec$ such that $(h, \prec)$ satisfies both \Crule{} and \Vrule{}) if and only if there exists a version order $\ll$ such that $\mathit{MVSG}(h, \ll)$ is acyclic.
\end{theorem}

Combined with \cref{thm:cv-rules}, this yields: $\text{\CVrules{}} \Leftrightarrow \text{MVSG Acyclicity} \Leftrightarrow \text{Serializable}$. All three characterizations identify the same set of histories. While \cref{thm:cv-mvsg-equivalence} could instead be obtained by composing \cref{thm:cv-rules} with the classical MVSG theorem (\cref{thm:mvsg}), we prove it directly: the construction exposes the structural duality between the two characterizations that we discuss in \cref{sec:why-cv-rules}.

\subsection{Proof Sketch}
\label{sec:equivalence-proof}

\textbf{Sufficiency} (\CVrules{} $\Rightarrow$ MVSG Acyclicity).
Given a transaction order satisfying \CVrules{}:
\begin{enumerate}
\item Extend the partial order to a total order using Szpilrajn's theorem~\cite{Szpi1930}.
\item Construct a version order $\ll$ by restricting this total order to writers of each data item.
\item Build the MVSG using this version order.
\item \textbf{Acyclicity proof}: Show that every MVSG edge respects the serial order.
\begin{itemize}
\item Initial edges: $t_0 \prec t_i$ holds since $t_0$ is the minimum element.
\item Reads-from edges: $t_j \prec t_i$ holds by \Crule{}.
\item Version-order and anti-dependency edges: By construction, the version order agrees with the serial order, so $t_k \ll_x t_j$ implies $t_k \prec t_j$ and $t_j \ll_x t_k$ implies $t_j \prec t_k$. A version-order edge $t_k \to t_j$ thus respects the serial order directly. For an anti-dependency edge $t_i \to t_k$ (which arises when $t_j \ll_x t_k$, hence $t_k \not\prec t_j$), \Vrule{} forces $t_i \prec t_k$, so this edge respects the serial order as well.
\end{itemize}
Since all edges respect the serial order (a total order), the graph is acyclic.
\end{enumerate}

\noindent\textbf{Necessity} (MVSG Acyclicity $\Rightarrow$ \CVrules{}).
Given an acyclic MVSG:
\begin{enumerate}
\item Extract a strict partial order from the acyclic graph. Since the graph is acyclic, the edge relation (and its transitive closure) forms a strict partial order on transactions.
\item \textbf{\Crule{} verification}: For each reads-from edge $t_j \to t_i$, the extracted order satisfies $t_j \prec t_i$.
\item \textbf{\Vrule{} verification}: For each read $\vf(x, t_i) = t_j$ and each other writer $t_k$, the version order places $t_k$ either before $t_j$ (yielding a version-order edge $t_k \to t_j$, so $t_k \prec t_j$) or after $t_j$ (yielding an anti-dependency edge $t_i \to t_k$, so $t_i \prec t_k$). Either case satisfies the \Vrule{} disjunction.
\end{enumerate}

The two directions are asymmetric: necessity extracts a partial order
directly from the acyclic graph, while sufficiency must first extend the
partial order to a total order. This extension is needed because MVSG
requires the version order for each data item to be total, whereas
\CVrules{} compare only the endpoints relevant to each read.

\subsection{\CVrules{} and MVSG Compared}
\label{sec:why-cv-rules}

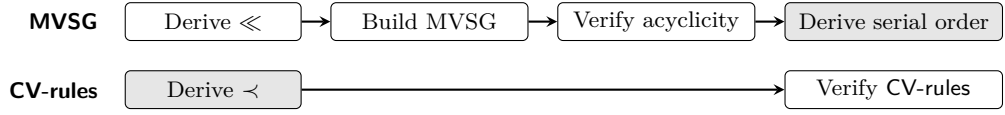
\begin{figure}[t]
\centering
\begin{tikzpicture}[
  box/.style={rectangle, draw, rounded corners=2pt,
              align=center, font=\footnotesize, inner sep=2.5pt,
              minimum height=0.5cm, text width=2.7cm},
  hdr/.style={font=\footnotesize\sffamily\bfseries, anchor=east},
  sbox/.style={box, text width=2.15cm},
  mbox/.style={box, text width=2.43cm},
  arr/.style={->, >=stealth, thick},
]

\def\row{0.9}   

\node[hdr] at (-0.6, 0) {MVSG};
\node[hdr] at (-0.6, -\row) {\CVrules{}};

\node[sbox] (L1) at (0.8, 0)      {Derive $\ll$};
\node[mbox] (L2) at (3.66, 0)     {Build MVSG};
\node[mbox] (L3) at (6.66, 0)     {Verify acyclicity};
\node[box, fill=black!10]  (L4) at (9.8, 0)      {Derive serial order};

\node[sbox, fill=black!10] (R1) at (0.8, -\row)  {Derive $\prec$};
\node[box]  (R2) at (9.8, -\row)  {Verify \CVrules{}};

\draw[arr] (L1) -- (L2);
\draw[arr] (L2) -- (L3);
\draw[arr] (L3) -- (L4);

\draw[arr] (R1) -- (R2);

\end{tikzpicture}
\caption{Two verification paths for serializability.
  MVSG constructs a graph and checks a global property (acyclicity),
  deriving the serial order as output.
  \CVrules{} verify local per-read conditions
  against an explicitly given transaction order.
  Both characterize the same set of serializable histories.}
\label{fig:mvsg-cv-comparison}
\end{figure}

Given that \CVrules{} and MVSG characterize the same histories, what advantages does each offer? The two approaches are dual in what they take as given and what they verify (\cref{fig:mvsg-cv-comparison}). MVSG builds a graph whose edges encode local relationships (reads-from, version-order, and anti-dependency) and verifies a global property: acyclicity of the entire graph. The serial order is derived from this graph as a byproduct. \CVrules{} take the opposite path: the transaction order is given as an explicit input, fixing the global structure, and verification reduces to checking local per-read conditions (\Crule{} and \Vrule{}). The indirect ordering and indirect proofs noted in \cref{sec:background} are both consequences of this structural asymmetry.

Which approach is more natural depends on how the protocol establishes serializability. Graph-based protocols such as IC3~\cite{Zhao2016-IC3}, Janus~\cite{Mush2016-Janus}, and Oze~\cite{Nemo2025-Oze} construct dependency graphs and check acyclicity at runtime; for these, MVSG is a natural foundation. However, many protocols establish serializability by constructing an order. For some, the order is explicit in the protocol mechanism, as with lock conflicts (2PL) or timestamps (MVTO); for others, such as SSN and Aria, the order is left implicit, and we recover it from logical timestamps and batch identifiers respectively. In either case \CVrules{} offer a more direct path than MVSG: with the order in hand, the proof checks per-read conditions against it.

Two case studies illustrate this contrast. For MVTO, both the MVSG-based proof~\cite{Bern1983} and our \CVrules{} proof rely on the same timestamp properties, but differ in how they are used: the MVSG proof routes them through version-order edges and graph acyclicity, while our \CVrules{} proof verifies \Crule{} and \Vrule{} directly against the timestamp order. For SSN, the difference is more fundamental. The original proof~\cite{Wang2017-SSN} proceeds by contradiction via an MVSG-like dependency graph, without ever constructing a serialization order. Our \CVrules{} proof instead identifies a logical timestamp not explicit in the original work and directly verifies both rules (see \cref{sec:ssn}).

\section{Protocol Verification}
\label{sec:case-studies}

We apply \CVrules{} to verify abstract models of five concurrency control protocols: Two-Phase Locking (2PL)~\cite{Eswa1976-2PL}, Multi-Version Timestamp Ordering (MVTO)~\cite{Reed1978-MVTO,Bern1983}, Serial Safety Net (SSN)~\cite{Wang2017-SSN}, Aria~\cite{Lu2020b-Aria}, and SnapChain. For each protocol, we establish \emph{soundness}: every execution satisfying the protocol's constraints is serializable. Each model considers only executions where all transactions satisfy the protocol's constraints; abort handling (e.g., deadlocks in 2PL, certification failures in SSN) is beyond the scope of this paper.

Soundness proofs follow a common pattern: we build an abstract model capturing the protocol's ordering and visibility constraints, construct a transaction order from the protocol's mechanisms, and verify \Crule{} and \Vrule{}. Each protocol, detailed below, illustrates a different aspect of \CVrules{}. 2PL derives a partial order from lock conflicts; MVTO uses predetermined timestamps as a total order. SSN derives a logical timestamp from dependency bounds, a construction not explicit in the original work. Aria combines batch ordering with dependency-based reordering within each batch; our proof also reveals that Aria's unique-write constraint is unnecessary for serializability. SnapChain reverses the direction, using \CVrules{} as a design principle: it enforces \Vrule{} by construction through per-data-item view chains.

Completeness for all five protocols holds, because they are serial-executable (\cref{prop:serial-executability}). The constructions are: 2PL acquires all locks in serial order; MVTO assigns timestamps according to serial-order position; SSN assigns commit timestamps as serial-order positions; Aria places each transaction in its own batch; SnapChain builds each data item's view chain following the serial order.

\subsection{Two-Phase Locking}
\label{sec:2pl}

\subsubsection{Protocol Model}

A 2PL execution consists of lock and unlock events with an ordering relation. Each transaction acquires locks before accessing data and releases them afterward. The key constraints are:
\begin{itemize}
\item \textbf{Lock-unlock pairing}: Each lock has a matching unlock.
\item \textbf{Two-phase discipline}: For each transaction, all locks precede all unlocks.
\item \textbf{Mutual exclusion}: Conflicting locks (same data, different transactions, at least one write) must be serialized: one transaction must unlock before the other locks.
\end{itemize}
We omit read/write events and derive the version function directly from the lock conflict order, selecting the maximal prior writer for each read. This abstraction is justified by the mutual exclusion invariant: since each read/write occurs between the corresponding lock and unlock, conflicting accesses to the same data item are serialized by the lock conflict order, which in turn determines the version function.

\subsubsection{Order Construction}

Lock conflicts naturally induce a transaction order. We define $t_j \prec t_i$ if $t_j$ unlocks data item $x$ before $t_i$ locks $x$, and the locks conflict (at least one is a write lock). The \emph{serialization order} is the transitive closure of this relation.

\subsubsection{Soundness}

\begin{theorem}[2PL Soundness]
Every 2PL execution is serializable.
\end{theorem}

The proof constructs an ordered history from the serialization order and verifies \CVrules{}:
\begin{itemize}
\item \textbf{\Crule{}}: If $t_i$ reads $x$ from $t_j$, then $t_j$ held a write lock on $x$ and released it before $t_i$ acquired its lock. Thus $t_j \prec t_i$.
\item \textbf{\Vrule{}} (via \cref{thm:v-rule-decomposition}): \emph{Endpoint comparability}: any other writer $t_k$ is ordered with both $t_j$ and $t_i$ by lock mutual exclusion. \emph{No interposition}: $t_j \prec t_k \prec t_i$ would contradict the maximality of $t_j$ as version source.
\end{itemize}

\subsection{Multi-Version Timestamp Ordering}
\label{sec:mvto}

\subsubsection{Protocol Model}

An MVTO execution maintains multiple versions of each data item and uses timestamps to determine serialization order~\cite{Reed1978-MVTO,Bern1983}. Each transaction $t$ receives a unique timestamp $\mathit{ts}(t)$ at the start of execution. The key components are:
\begin{itemize}
\item \textbf{Version store}: Each data item has multiple versions, each tagged with its writer $t$'s timestamp $\mathit{ts}(t)$ as \textsf{writeTs}.
\item \textbf{Version selection}: A read by transaction $t_i$ selects the version with the largest \textsf{writeTs} less than $\mathit{ts}(t_i)$.
\item \textbf{Timestamp uniqueness}: Different transactions have different timestamps; $t_0$ has the minimum timestamp $\mathit{ts}(t_0) = 0$.
\end{itemize}
Unlike 2PL, MVTO requires no locks; concurrency is achieved through multi-versioning.

\subsubsection{Order Construction}

The serialization order is simply the timestamp order: $t_j \prec t_i$ iff $\mathit{ts}(t_j) < \mathit{ts}(t_i)$. This is a total order determined at transaction start, before any data access occurs.

\subsubsection{Soundness}

\begin{theorem}[MVTO Soundness]
Every MVTO execution is serializable.
\end{theorem}

The proof constructs an ordered history from the timestamp order and verifies \CVrules{}:
\begin{itemize}
\item \textbf{\Crule{}}: If $t_i$ reads $x$ from $t_j$, then $t_i$ selected a version with $\textsf{writeTs} = \mathit{ts}(t_j)$. Version selection requires $\textsf{writeTs} < \mathit{ts}(t_i)$, so $\mathit{ts}(t_j) < \mathit{ts}(t_i)$, establishing $t_j \prec t_i$.
\item \textbf{\Vrule{}} (via \cref{thm:v-rule-decomposition}): \emph{Endpoint comparability} is immediate since timestamps form a total order. \emph{No interposition}: if $\mathit{ts}(t_j) < \mathit{ts}(t_k) < \mathit{ts}(t_i)$, then $t_k$'s version would have been selected instead of $t_j$'s, a contradiction.
\end{itemize}

\subsection{Serial Safety Net}
\label{sec:ssn}

\subsubsection{Protocol Model}

SSN~\cite{Wang2017-SSN} is a \emph{verification} protocol: rather than scheduling operations, it checks whether a given execution is serializable. Each transaction receives a unique physical commit timestamp $\mathit{ct}(t)$. Since our model considers only committed transactions, $\mathit{ct}$ is simply a given injective assignment from transactions to timestamps. The protocol requires that each transaction reads only from already-committed transactions, i.e., $\mathit{ct}(t_j) < \mathit{ct}(t_i)$ whenever $t_i$ reads from $t_j$. SSN verifies serializability by analyzing dependencies.

An \emph{anti-dependency} from $t_i$ to $t_j$ exists when $t_i$ reads a data item $x$ and $t_j$ later writes a newer version of $x$ in commit-timestamp order (i.e., $t_j$ overwrites the version $t_i$ read). For each transaction $t$, we define two sets that capture the ordering constraints SSN enforces:
\begin{itemize}
\item $T^{\mathit{forward}}_t$ (must precede $t$):
  (a)~transactions that $t$ reads from;
  (b)~transactions whose version of some item is overwritten by~$t$; and
  (c)~transactions from which there is a direct (single-step) anti-dependency to~$t$ and that committed before~$t$.
\item $T^{\mathit{backward}}_t$ (must follow $t$): transactions reachable from $t$ via anti-dependencies, restricted to those that committed before $t$.
\end{itemize}
From these sets, SSN's certification values are defined as:
\begin{itemize}
\item $\eta(t) = \max\{\mathit{ct}(t') \mid t' \in T^{\mathit{forward}}_t\}$ (\emph{forward constraint}).
\item $\pi(t) = \min(\{\mathit{ct}(t') \mid t' \in T^{\mathit{backward}}_t\} \cup \{\mathit{ct}(t)\})$ (\emph{backward constraint}).
\end{itemize}
An execution is \emph{SSN-certified} if $\eta(t) < \pi(t)$ for all transactions $t$. Intuitively, $\eta(t)$ is the latest commit timestamp that must precede $t$ and $\pi(t)$ the earliest that must follow it; certification $\eta(t) < \pi(t)$ means these two windows are disjoint, so $t$ can be consistently placed between its predecessors and successors.

\subsubsection{Order Construction}

The original SSN paper proves serializability indirectly, by showing that the certification condition prevents cycles in a serialization graph, but does not explicitly construct a serialization order. We show that SSN implicitly defines a \emph{logical timestamp} $\rho(t)$ that determines the serialization order.

We construct the order from $\pi(t)$ values. However, multiple transactions may share the same $\pi$ value, so $\pi$ alone does not suffice. The key observation for tie-breaking is: when $\pi(t) = \pi(t')$ and $\mathit{ct}(t) < \mathit{ct}(t')$, an anti-dependency $t' \to t$ connects them, requiring $t'$ (which committed later but read a version that $t$ overwrote) to \emph{precede} $t$ in the serialization order. That is, among transactions with the same $\pi$, the physical commit order must be \emph{reversed}.

This motivates the definition $\rho(t) = (\pi(t), -\mathit{ct}(t))$ with lexicographic order (\cref{fig:ssn-rho}), and the serialization order is $t \prec t'$ iff $\rho(t) < \rho(t')$. We call $\rho$ a ``logical timestamp'' because it differs from the physical commit order. The SSN-certified condition $\eta(t) < \pi(t)$ ensures that transactions preceding $t$ (which contribute to $\eta(t)$) have $\rho$ values less than $\rho(t)$, since $\rho(t)_1 = \pi(t) > \eta(t)$. A topological sort of the MVSG would also yield some linear extension of the dependency graph, but without explicit correspondence to SSN's runtime variables. Instead, $\rho$ yields the order directly from the protocol's certification values.

\begin{figure}[t]
\centering
\begin{tikzpicture}[
  dot/.style={circle, fill, inner sep=1.5pt},
  lbl/.style={font=\footnotesize},
  arr/.style={->, semithick},
  anti/.style={->, semithick, densely dashed},
  num/.style={circle, draw, inner sep=0.5pt, minimum size=11pt,
              font=\scriptsize, fill=white},
]

\draw[arr] (0, 0) -- (8.2, 0) node[right, lbl] {$\pi$};
\draw[arr] (0, 0) -- (0, 3.9) node[above, lbl] {$\mathit{ct}$};

\draw[gray!40, dashed] (2, 0) -- (2, 3.5);
\draw[gray!40, dashed] (5.5, 0) -- (5.5, 3.5);
\node[lbl, below=2pt] at (2, 0) {$\pi_1$};
\node[lbl, below=2pt] at (5.5, 0) {$\pi_2$};

\node[dot] (u) at (2, 1.5) {};
\node[lbl, left=3pt] at (u) {$u$};
\node[num] at (2.32, 1.82) {1};

\node[dot] (tp) at (5.5, 3.0) {};
\node[lbl, right=3pt] at (tp) {$t'$};
\node[num] at (5.18, 3.26) {2};

\node[dot] (tt) at (5.5, 1.1) {};
\node[lbl, right=3pt] at (tt) {$t$};
\node[num] at (5.18, 0.84) {3};

\draw[anti] (tp) to[bend left=40]
  node[right, align=left, font=\footnotesize, xshift=1pt] {backward\\anti-dep.} (tt);

\draw[arr] (1.2, -0.95) -- (6.3, -0.95)
  node[midway, below, lbl] {serialization order $\prec$};

\end{tikzpicture}
\caption{Constructing the serialization order from
  $\rho(t) = (\pi(t), -\mathit{ct}(t))$.
  Transactions are ordered primarily by $\pi$ (the backward constraint);
  the numbered circles give the resulting order $\prec$.
  Within one $\pi$ value (here $t, t'$ at $\pi_2$),
  the tie-break $-\mathit{ct}$ places the larger commit timestamp first.
  This realizes a \emph{backward} anti-dependency $t' \to t$:
  $t'$ read a version that $t$ overwrote, yet $\mathit{ct}(t) < \mathit{ct}(t')$,
  so $t \in T^{\mathit{backward}}_{t'}$.
  The reader~$t'$ must precede the overwriter~$t$,
  so the transaction that committed later is serialized earlier.}
\label{fig:ssn-rho}
\end{figure}
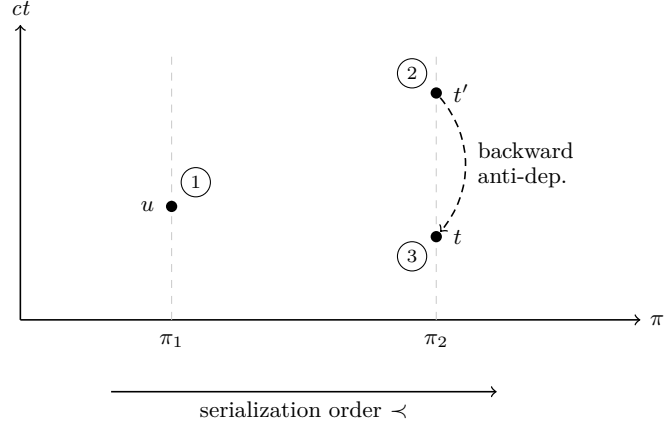

\subsubsection{Soundness}

\begin{theorem}[SSN Soundness]
Every SSN-certified execution is serializable.
\end{theorem}

The proof verifies \CVrules{} using the logical timestamp order:
\begin{itemize}
\item \textbf{\Crule{}}: If $t_i$ reads from $t_j$, then $t_j \in T^{\mathit{forward}}_{t_i}$, so $\mathit{ct}(t_j) \le \eta(t_i)$. Chaining inequalities: $\rho(t_j)_1 = \pi(t_j) \le \mathit{ct}(t_j) \le \eta(t_i) < \pi(t_i) = \rho(t_i)_1$, hence $\rho(t_j) < \rho(t_i)$.
\item \textbf{\Vrule{}} (via \cref{thm:v-rule-decomposition}): \emph{Endpoint comparability} holds because the logical timestamp order is total. \emph{No interposition}: suppose for contradiction that $t_j \prec t_k \prec t_i$, i.e., $\rho(t_j) < \rho(t_k) < \rho(t_i)$. First, $\mathit{ct}(t_j) < \mathit{ct}(t_k)$ must hold; otherwise $\rho(t_k) < \rho(t_j)$ by the same reasoning as \Crule{} (since $t_k \in T^{\mathit{forward}}_{t_j}$), contradicting $\rho(t_j) < \rho(t_k)$. Then $t_k$ writes a newer version of~$x$ than~$t_j$, creating an anti-dependency from $t_i$ to $t_k$. Two cases arise (by $\mathit{ct}$ uniqueness):
  \begin{enumerate}
  \item $\mathit{ct}(t_i) < \mathit{ct}(t_k)$: then $t_i \in T^{\mathit{forward}}_{t_k}$, so $\pi(t_i) \le \mathit{ct}(t_i) \le \eta(t_k)$ and hence $\rho(t_i) < \rho(t_k)$, a contradiction.
  \item $\mathit{ct}(t_k) < \mathit{ct}(t_i)$: then $t_k \in T^{\mathit{backward}}_{t_i}$. By transitivity of anti-dependencies, $T^{\mathit{backward}}_{t_k} \cup \{t_k\} \subseteq T^{\mathit{backward}}_{t_i}$, so $\pi(t_i) \le \pi(t_k)$; combined with $\mathit{ct}(t_k) < \mathit{ct}(t_i)$, this gives $\rho(t_i) < \rho(t_k)$, a contradiction.
  \end{enumerate}
\end{itemize}

\subsection{Aria}
\label{sec:aria}

\subsubsection{Protocol Model}

Aria~\cite{Lu2020b-Aria} is a deterministic concurrency control protocol based on batch execution. Transactions are organized into an ordered sequence of batches $B_0, B_1, \ldots$. Within each batch, transactions receive unique identifiers called \emph{Aria-tids}, denoted $\mathit{atid}(t)$, and their serialization order is determined by dependency analysis.

For transaction $t_i$ in batch $B_n$, we define the \emph{predecessor set} $T^{\mathit{prev}}_i = \{t_j \in B_n \mid \mathit{atid}(t_j) < \mathit{atid}(t_i)\}$. The key constraints are:
\begin{itemize}
\item \textbf{Batch snapshot}: Each batch reads from a snapshot determined by preceding batches. Batch $B_0$ reads from $t_0$; for $n > 0$, batch $B_n$ reads from the state produced by all committed writes in batches $B_0, \ldots, B_{n-1}$, selecting, for each data item, the $\prec$-maximum writer among the preceding batches under the order constructed in \cref{sec:aria-order}. (This involves only the restriction of~$\prec$ to $B_0, \ldots, B_{n-1}$, a total order that does not depend on $B_n$'s own intra-batch order.)
\item \textbf{Dependency exclusivity}: Within each batch, each transaction has at most one type of dependency with its predecessors. A transaction $t_i$ has a \emph{WAR dependency} (write-after-read) if $t_i$ writes a data item that some $t_j \in T^{\mathit{prev}}_i$ reads, and a \emph{RAW dependency} (read-after-write) if some $t_j \in T^{\mathit{prev}}_i$ writes a data item that $t_i$ reads. Aria requires that no transaction has both types simultaneously.
\end{itemize}
The batch snapshot property ensures all transactions in a batch see a consistent state, and dependency exclusivity enables deterministic reordering within each batch. The original Aria protocol also requires unique writes per data item (no WAW dependency), but this constraint is unnecessary for serializability: reads are fully determined by the batch snapshot, and the intra-batch serialization order is fully determined by dependency exclusivity, regardless of write overlaps. Our proof never uses the unique-write assumption. In Aria's original setting, the no-WAW constraint avoids parallel installations to the same record in its runtime-level concrete model, a concern outside our version-function-only model.

\subsubsection{Order Construction}
\label{sec:aria-order}

The serialization order combines batch ordering with dependency-based reordering within each batch:
\begin{itemize}
\item $t_0 \prec t$ for all $t \neq t_0$.
\item \textbf{Inter-batch}: If $t_i \in B_m$ and $t_k \in B_n$ with $m < n$, then $t_i \prec t_k$.
\item \textbf{Intra-batch}: Within each batch, the order depends on dependency type:
  \begin{itemize}
  \item If $t$ has a RAW dependency, then $t \prec t'$ for all $t' \in T^{\mathit{prev}}_t$ (reorder $t$ before its predecessors).
  \item If $t$ has no RAW dependency, then $t' \prec t$ for all $t' \in T^{\mathit{prev}}_t$ (keep predecessors before $t$).
  \end{itemize}
\end{itemize}
Dependency exclusivity ensures the intra-batch order construction is well-defined. For any pair $(t_i, t_j)$ in the same batch with $\mathit{atid}(t_i) < \mathit{atid}(t_j)$, the ordering depends solely on whether $t_j$ has a RAW dependency: if so, $t_j \prec t_i$; otherwise, $t_i \prec t_j$. In particular, if $t_i$ reads and $t_k$ writes the same data item within a batch, then $t_i \prec t_k$ regardless of Aria-tid order.

\subsubsection{Soundness}

\begin{theorem}[Aria Soundness]
Every Aria execution is serializable.
\end{theorem}

The proof verifies \CVrules{} using this serialization order:
\begin{itemize}
\item \textbf{\Crule{}}: If $t_i$ in batch $B_n$ reads $x$ from $t_j$, then $t_j$ is the writer selected by the batch snapshot of $B_n$ at $x$. If $t_j = t_0$, then $t_0 \prec t_i$ by the base rule. Otherwise, $t_j$ belongs to some earlier batch $B_m$ with $m < n$, so $t_j \prec t_i$ by inter-batch ordering.
\item \textbf{\Vrule{}}: Consider $t_i$ in batch $B_n$ reading $x$ from $t_j$, and another writer $t_k$ of $x$. If $t_k$ is in a batch before $B_n$, then $t_k \prec t_j$ because the snapshot selects the maximum writer and $t_k \ne t_j$. Otherwise $t_k$ is in $B_n$ or later, giving $t_i \prec t_k$: by intra-batch ordering (dependency exclusivity ensures readers precede writers of the same item) if $t_k \in B_n$, or by inter-batch ordering if $t_k$ is in a later batch.
\end{itemize}

\subsection{SnapChain}
\label{sec:snapchain}

\subsubsection{Protocol Model}

The preceding four protocols were designed independently and verified against \CVrules{} after the fact. SnapChain demonstrates the reverse direction: using \CVrules{} as a starting point for design. SnapChain uses per-data-item \emph{view chains} (ordered slots of writers and readers) to enforce \Vrule{} by construction.

For each data item $x$, a \emph{view chain} is a list of \emph{slots} $[(W_0, R_0, M_0), (W_1, R_1, M_1), \ldots]$.
Each slot corresponds to a version of $x$: the writers produce the version, the readers consume it, and a read-modify-write (rmw) transaction both consumes the current version and produces the next.
\begin{itemize}
\item $W_s$ (\emph{writers}): a non-empty set of transactions that write $x$ at this position.
\item $R_s$ (\emph{readers}): transactions that read $x$ from the writer of this slot without writing $x$.
\item $M_s$ (\emph{rmw}): a set of at most one transaction that reads from this slot's writer and writes the next version of $x$.
\end{itemize}
The initial slot has $W_0 = \{t_0\}$. If a slot has readers or an rmw ($R_s \cup M_s \neq \emptyset$), then $|W_s| = 1$: readers see a unique version. If $M_s = \{t\}$, then slot $s+1$ exists with $W_{s+1} = \{t\}$.

Each transaction $t$ has start and end events with a partial order: $\mathit{end}(t_0) < \mathit{start}(t)$ for all $t \neq t_0$, and $\mathit{start}(t) < \mathit{end}(t)$.
A transaction can only observe effects of transactions that ended before it started. This determines its snapshot.
Each transaction $t \neq t_0$ that accesses $x$ is placed in the view chain based on its unique \emph{visible slot}: the latest slot containing a member $t'$ that accesses $x$ and satisfies $\mathit{end}(t') < \mathit{start}(t)$.
Read-only transactions join $R_s$, blind writers join $W_{s+1}$, and rmw transactions join both $M_s$ and $W_{s+1}$.
The version function is derived from the chain: if $t_i \in R_s \cup M_s$, it reads from the unique writer in $W_s$.

In practice, determining the visible slot requires observing which transactions have completed before a new transaction starts. Our model takes this information as given.

\subsubsection{Order Construction}

The serialization order is built from per-data-item constraints. For each data item $x$, the \emph{per-data order} is generated by:
\begin{itemize}
\item \textbf{Intra-slot}: $W_s \prec R_s \prec M_s$.
\item \textbf{Cross-slot}: When $M_s = \emptyset$, $W_s \cup R_s \prec W_{s+1}$. When $M_s = \{t\}$, $t$ bridges slot $s$ and $s+1$ as it appears in both $M_s$ and $W_{s+1}$.
\end{itemize}
The \emph{serialization order} is the transitive closure of the union of all per-data orders. This requires that per-data orders from different data items are mutually consistent; when they are not, some transactions must be aborted to restore consistency. The model assumes that committed transactions generate a consistent global order; abort strategies are outside the scope.

\subsubsection{Soundness}

\begin{theorem}[SnapChain Soundness]
Every SnapChain execution is serializable.
\end{theorem}

The proof verifies \CVrules{} using the serialization order:
\begin{itemize}
\item \textbf{\Crule{}}: $t_j \in W_s$ and $t_i \in R_s \cup M_s$ for some slot $s$; the intra-slot ordering gives $t_j \prec t_i$.
\item \textbf{\Vrule{}}: Since $|W_s| = 1$, $t_k$ appears in slot $s' \neq s$. If $s' < s$, the slot ordering gives $t_k \prec t_j$; if $s < s'$, it gives $t_i \prec t_k$.
\end{itemize}

\section{Mechanization in Lean}
\label{sec:mechanization}

All theorems in this paper are mechanically verified in Lean (version 4.31.0).
The artifact, including build instructions and a mapping from paper theorems to Lean definitions, is available as supplementary material~\cite{Hosh2026-CVRules-Artifact}.
The proof development comprises about 11.2K lines of code across 105 modules, with zero axioms and zero admitted goals. ``Zero axioms'' means our development adds no axioms beyond those already provided by Lean's standard library; there are no \texttt{axiom} declarations of our own. ``Zero admitted goals'' means every proof obligation is fully discharged; there are no gaps marked with \texttt{sorry}. Together, these properties ensure that anyone can independently verify the correctness of our theorems by running the Lean type checker.

\subsection{Proof Statistics}
\label{sec:proof-statistics}

\cref{tab:metrics} summarizes the proof development. The core theory includes \CVrules{} definitions, the characterization theorem, order-theoretic lemmas, and version function correctness. The bounded-width component formalizes Dilworth's theorem and the state-space search decision procedure. Protocol implementations vary in complexity: MVTO is the largest because it explicitly formalizes version selection from a multi-version store, while Aria is the smallest.

\begin{table}[t]
\centering
\small
\caption{Proof development metrics (Defs = definitions, LOC = lines of code)}
\label{tab:metrics}
\begin{tabulary}{\linewidth}{Lrrrr}
\hline
Component & Modules & Defs & Lemmas & LOC(\%) \\
\hline
Core Theory & 16 & 49 & 114 & 2,145(\phantom{0}19.2) \\
MVSG Equivalence & 3 & 18 & 11 & 538(\phantom{00}4.8) \\
Bounded Width & 9 & 29 & 69 & 1,579(\phantom{0}14.2) \\
Protocol Common & 4 & 3 & 15 & 150(\phantom{00}1.3) \\
2PL & 13 & 25 & 48 & 1,220(\phantom{0}10.9) \\
MVTO & 22 & 23 & 88 & 2,152(\phantom{0}19.3) \\
SSN & 14 & 24 & 55 & 1,120(\phantom{0}10.0) \\
Aria & 10 & 24 & 48 & 991(\phantom{00}8.9) \\
SnapChain & 14 & 32 & 58 & 1,258(\phantom{0}11.3) \\
\hline
\textbf{Total} & \textbf{105} & \textbf{227} & \textbf{506} & \textbf{11,153}\textnormal{(100.0)} \\
\hline
\end{tabulary}
\end{table}

\subsection{Implementation Notes}
\label{sec:implementation-notes}

A notable implementation detail is representing SSN's logical timestamps $\rho(t) = (\pi(t), -\mathit{ct}(t))$ as a lexicographic pair of natural numbers with the order on the second component reversed. Keeping both components in~$\mathbb{N}$ avoids non-negativity invariants and lets standard arithmetic tactics apply directly, which simplifies the proofs. The injectivity of~$\rho$ (needed to ensure the serialization order is well-defined) then reduces to the injectivity of commit timestamps.

We adopted the decomposed \Vrule{} approach (\cref{thm:v-rule-decomposition}) for 2PL, MVTO, and SSN. For Aria, the batch position of the competing writer relative to the reader and the snapshot source uniquely determines which disjunct of \Vrule{} holds in each case, making decomposition unnecessary. For 2PL and MVTO, the decomposition automatically determines the case splits but increases the number of cases compared to a non-decomposed proof; however, the additional cases are trivial (five cases for 2PL, four for MVTO, each a few lines), so the overhead is negligible. For SSN, the decomposition conversely reduces the number of cases from six to four, because the hypotheses of the no-interposition obligation cover and merge case distinctions that a non-decomposed proof must handle separately. SnapChain instead verifies \Vrule{} directly from its view-chain slot structure, without the decomposition.

The bounded-width tractability result (\cref{sec:bounded-width}) required formalizing Dilworth's theorem and the entire state-space search pipeline as a verified decision procedure. The mechanization establishes that reachability in the state space is equivalent to serializability (soundness and completeness) and that reachability is decidable, with search termination proved via a monotonically decreasing measure on unplaced transactions. The complexity claims are informal, since formalizing complexity bounds in Lean remains a significant undertaking.

\section{Related Work}
\label{sec:related-work}

\subsection{Serializability Characterizations}
The classical theory of serializability developed from single-version to
multi-version settings, and from conflict-based to view-based definitions.
For single-version histories, conflict serializability, characterized by
acyclicity of the serialization graph~\cite{Eswa1976-2PL}, is decidable in polynomial
time, while view serializability is
NP-complete~\cite{Papa1979}.
Bernstein et al.~\cite{Bern1983,Bern1987-book} extended these ideas to multi-version
histories with the Multi-Version Serialization Graph (MVSG), where a history
is serializable if and only if some version order makes the MVSG acyclic.
These graph-based approaches characterize serializability via acyclicity of a dependency graph, without requiring a serialization order as input.
Biswas and Enea~\cite{Bisw2019} systematically studied the complexity of checking consistency models including serializability, showing that it is decidable in polynomial time for any fixed bound on the number of sessions.

More recent frameworks adopt declarative specifications.
Cerone et al.~\cite{Cero2015} proposed an axiomatic framework using
visibility (VIS) and arbitration (AR) relations, defining a hierarchy
from Read Atomic to Serializability. AR is a total order at all levels;
VIS varies by isolation level, becoming total and equal to AR at
serializability.
Cerone and Gotsman~\cite{Cero2016} extended this to characterize
Snapshot Isolation via dependency graphs, with applications to
transaction chopping and robustness~\cite{Bern2016}.
Crooks et al.~\cite{Croo2017-Seeing} introduced state-based definitions where
executions form a \emph{totally ordered sequence of states}.
Xiong et al.~\cite{Xion2020} developed operational semantics for
distributed key-value stores with execution tests, where the derived
transaction order can be partial; however, their model assumes
last-write-wins (LWW), making it difficult to model protocols like
MVTO (timestamp-based version selection) or Aria (batch-snapshot-based reading).

Our \CVrules{} characterization differs in two key aspects:
(1)~we work with \emph{partial orders} rather than total orders,
unlike Cerone's framework where AR's totality anchors the most-recent-writer rule of its external-consistency axiom (EXT);
\Vrule{}'s disjunctive structure instead localizes comparability to each read-from edge,
allowing protocols like 2PL to use their natural lock conflict order directly;
(2)~we provide \emph{per-read conditions} (verifying \Crule{} and \Vrule{}
per read-from edge) with an explicit transaction order, rather than
graph acyclicity conditions~\cite{Bern1983,Cero2016}.
Moreover, the C/V separation reveals the computational structure of serializability.

\subsection{Protocol Verification}
Concurrency control protocols can be classified by how they construct
serialization order: \emph{lock-based} protocols (2PL~\cite{Eswa1976-2PL}, Silo~\cite{Tu2013-Silo}, Orthrus~\cite{Renk2016-Orthrus}, 2PLSF~\cite{Rama2023-2PLSF}) derive
order from lock conflicts; \emph{timestamp-based} protocols (MVTO~\cite{Reed1978-MVTO}, TicToc~\cite{Yu2016-TicToc},
Cicada~\cite{Lim2017-Cicada}) use timestamps directly; \emph{graph-based} protocols (IC3~\cite{Zhao2016-IC3}, Janus~\cite{Mush2016-Janus}, Oze~\cite{Nemo2025-Oze})
infer order via topological sort of acyclic dependency graphs.
SSN~\cite{Wang2017-SSN} and Aria~\cite{Lu2020b-Aria} are notably different: they define conditions for
committing transactions without explicit order construction, leaving
the underlying serialization order implicit.
Classical protocols like 2PL and MVTO have well-established correctness
proofs in textbooks~\cite{Bern1987-book,Weik2002-book}.

To our knowledge, this paper provides the first explicit order constructions for SSN and Aria, within a unified framework that also covers 2PL, MVTO, and SnapChain.

\subsection{Formal and Automated Verification}
Mechanized verification of concurrency control varies in
\emph{what} is verified and \emph{how} serializability is established.
\emph{Implementation verification:}
vMVCC~\cite{Chan2023} and DaisyNFS~\cite{Chaj2022} prove that concrete
implementations correctly realize a transactional specification
using Iris/Perennial in Coq; these proofs are protocol-specific
and tightly coupled to implementation details.
\emph{Abstract protocol models:}
Chkliaev et al.~\cite{Chkl1999} verified 2PL in PVS via
conflict-preserving timestamps;
Lesani et al.~\cite{Lesa2012} and Doherty et al.~\cite{Dohe2013}
verified STM in PVS using I/O automata, requiring \emph{total order}
for committed transactions.
C4~\cite{Lesa2022-C4} formulates serializability as linearizability;
VerIso~\cite{Ghas2025} proves refinement to abstract serial schedules in Isabelle/HOL,
discovering a bug in TAPIR~\cite{Zhan2018-TAPIR},
though it inherits LWW limitations from Xiong et al.~\cite{Xion2020}.
Mathiasen et al.~\cite{Math2025} specify weak isolation levels
for client reasoning rather than characterizing serializability.
\emph{Runtime checking:}
Cobra~\cite{Tan2020a-Cobra}, Elle~\cite{King2020-Elle}, and their
successors~\cite{Gulo2024-IsoVista,Sunw2025-Vbox} check recorded traces
post hoc rather than proving correctness for all executions.

\CVrules{} directly characterize serializability, as both
necessary and sufficient conditions, via per-read conditions on
explicit partial orders. They differ from the approaches above
in three respects: they abstract away implementation details
(cf.\ vMVCC, DaisyNFS); they prove correctness for all executions
rather than verifying individual ones (cf.\ Cobra, Elle); and
they require neither total orders~\cite{Lesa2012,Dohe2013} nor
LWW semantics~\cite{Xion2020}.

\section{Conclusion}
\label{sec:conclusion}

We presented \CVrules{}, a characterization of serializability via two
local per-read conditions (\Crule{} and \Vrule{}) on an explicit
transaction order, and proved their equivalence with MVSG acyclicity.
We also showed that any fixed bound on the width of the order forced by \Crule{} yields polynomial-time decidability.
Using this framework, we verified 2PL, MVTO, SSN, Aria, and SnapChain,
identifying explicit order constructions for SSN and Aria
whose original papers provided only certification conditions.
Furthermore, we found that Aria's unique-write constraint is unnecessary for serializability,
and we used \CVrules{} as a design principle for SnapChain.
All results except the complexity bounds are mechanized in Lean (${\sim}11.2\text{K}$ lines).

Future work includes applying \CVrules{} as a design principle to develop protocols with greater scheduling power and lower overhead.
Another direction is bridging the gap between scheduling power and deployment performance by extending the verification framework to incorporate abort handling and workload models.

\bibliography{references}


\end{document}